\documentclass[letterpaper, 10 pt, conference]{ieeeconf}  

\IEEEoverridecommandlockouts                              

\usepackage{graphicx} 
\usepackage{amsmath,amssymb,amsfonts}
\usepackage{algorithm}
\usepackage{algorithmic}
\usepackage{booktabs}
\usepackage{graphicx}    
\usepackage{subcaption}  
\usepackage[capitalise]{cleveref}
\usepackage{comment}
\usepackage{cite}
\newtheorem{corollary}{Corollary}
\newtheorem{remark}{Remark}
\allowdisplaybreaks

\title{\LARGE \bf
Real-Time In-Domain Congestion Control for the LWR Traffic Model via Control Barrier Functions
}

\author{Zehang Zhu$^{1}$, Brian Block$^{1}$ and Stephanie Stockar$^{1}$
\thanks{This material is based upon work supported by the National Science Foundation CAREER Award 2042354.}
\thanks{$^{1}$Authors are with the Department of Mechanical and Aerospace Engineering and the Center for Automotive Research, The Ohio State University, 930 Kinnear Road, Columbus, OH 43212, USA {\tt\small \{zhu.4178, block.168, stockar.1\}@osu.edu}}%
}

\begin{document}

\maketitle

\thispagestyle{empty}
\pagestyle{empty}

\begin{abstract}
This paper presents a control barrier function-based method for real-time in-domain congestion control of the Lighthill-Whitham-Richards traffic model. Traffic congestion is formulated as a distributed safety control problem, leading to a infinite-dimensional optimization problem. Through the discretization and Karush-Kuhn-Tucker (KKT) analysis, the problem is converted into a high-dimensional quadratic program. A structure-exploiting primal-dual active set algorithm is then developed to compute the safe control input in real time, with convergence guarantees. Numerical simulations with different nominal controllers demonstrate the effectiveness and real-time feasibility of the proposed approach.
\end{abstract}

\section{Introduction}
In recent years, mobility demand has increased continuously, causing an increase in traffic congestion. From 2023 to 2024, the average time spent in congestion rose by $20\%$ \cite{FHWA2024Urban}. 
Advanced traffic control strategies, supported by rapid developments in connected transportation technologies, offer a promising means of reducing congestion and improving traffic efficiency and safety \cite{Siri2021Freeway}. 

Macroscopic partial differential equation (PDE) models provide an efficient framework to model the large-scale traffic dynamics that lead to congestion and can support control design.
A commonly used macroscopic model is the Lighthill–Whitham–Richards (LWR) model, which is a first-order, hyperbolic PDE that captures the spatiotemporal evolution of traffic density. The LWR model assumes that the velocity of vehicles adapts instantaneously based on a prescribed fundamental diagram \cite{Lighthill1955Kinematic, Richards1956Shock}. 
Highlighted for its simplicity and effectiveness, the LWR model is commonly used in traffic control, such as in boundary ramp-metering control \cite{Tumash2019Robust, Yu2021Bilateral}, moving-bottleneck control implemented via connected and automated vehicles \cite{Block2025Analysis}, or in-domain variable speed limit (VSL) control \cite{Block2024LQ}.

With in-domain control, the actuator can have complete control over the internal traffic dynamics. One way to achieve the in-domain control for traffic models is to implement VSLs. In \cite{Carlson2011Local}, VSLs were implemented using a PI-type feedback on a discrete model to control flow rate and traffic density. Model predictive control for VSLs was used in \cite{Liu2017Model} to reduce total time spent in traffic and emissions. The use of VSL actuation to track a desired outflow was also developed in \cite{DelleMonache2017Traffic}. In \cite{Block2024LQ}, a continuous linear quadratic (LQ) controller was proposed to stabilize the LWR model through VSLs. These works, though, focus only on stabilizing the system to a desired setpoint or trajectory, not on ensuring system states stay within a desired set.

Maintaining the traffic density in free-flow can be interpreted as a distributed safety control problem, which motivates the use of control barrier functions (CBFs), which ensure the system's safety by enforcing the invariance of a safe set of states. The use of CBFs in ODE systems has been well studied. In \cite{Ames2019Control}, a CBF-based quadratic programming (QP) solver was proposed for control-affine ODE systems to solve the safe control input efficiently. The effectiveness of CBF-based methods has been verified in various microscopic driving scenarios \cite{Alan2023Control, Liu2023SafetyCritical, Cho2023Model}, where the CBF guarantees the driving safety of autonomous vehicles. The use of CBFs for control in PDE systems has been less studied, though, and mainly for boundary control, not in-domain control. In \cite{Koga2023Safe}, a second-order CBF was employed in the Stefan model at a boundary to prevent liquid freezing while driving the phase interface to a desired location. This approach was further extended to third-order moving-boundary dynamics in \cite{Koga2025Safe}. The safety control problem of an unmanned aerial vehicle with hanging loads, formulated as an ODE-PDE-ODE system, was discussed in \cite{Wang2024Safe}, in which a CBF was used to regulate the input to ensure the safety of the loads. As mentioned, these works focus on boundary control, and the in-domain PDE safety control problem has been less studied in the literature. In particular, existing approaches do not address real-time computation of in-domain safe controls for traffic PDEs.  To address the current limitations in safe, in-domain control of PDEs, we propose a real-time solvable, CBF-based method with VSLs to prevent congestion for the LWR model. The main contributions are as follows:
(i)A formulation of CBF-based in-domain safe control for the LWR traffic model, addressing a gap in existing CBF approaches that focus primarily on boundary control;
(ii) A structure-exploiting primal-dual active set (PDAS) algorithm for real-time computation of the optimal safe control input, with convergence guarantees;
(iii) A numerical validation demonstrating safety enforcement and real-time feasibility under different nominal controllers.


\section{Model and Problem Description}
\label{Section Model}

\subsection{LWR Model}
We consider the LWR model with VSL in-domain control on a straight road of length $L$ with on-ramps, where position $x\in[0,L]$ and time $t\in[0,T]$. The model is presented as
\begin{equation}
    \frac{\partial\rho}{\partial t} + \frac{\partial (uQ)}{\partial x} = q
    \label{LWR model}
\end{equation}
where $\rho(x,t): [0, L]\times[0, T]\rightarrow\mathbb{R}$ is the traffic density; $q(x,t)$ is the on-ramp flow; and $u(x,t)$ is the continuous VSL control input. $Q(x,t)$ is the traffic flow described by
\begin{equation}
    Q = f(\rho)
    \label{FD}
\end{equation}
where $f(\rho): [0,\rho_{\text{max}}]\rightarrow[0, Q_{\text{max}}]$ gives the equilibrium relationship between traffic density and flow. The maximum density, flow, and speed are $\rho_{\text{max}}, Q_{\text{max}}$, and $V_{\text{max}}$, respectively. The corresponding maximum flow
\begin{equation}
    Q_{\text{max}} = f(\rho_c)\geq f(\rho), \forall \rho \in [0,\rho_{\text{max}}]
    \label{critical}
\end{equation}
happens at the critical density $\rho_c$, where the traffic state transitions from free-flow to congestion. 
The boundary and initial conditions of \eqref{LWR model} are 
\begin{equation}
    Q(0,t) = Q_{\text{{up}}}(t),\ \rho(x,0) = \rho^0(x)
\end{equation}
where $Q_{\text{{up}}}(t)$ is the traffic flow from the upstream, and we assume that the supply of the downstream is big enough; $\rho^0(x)$ is the initial density on the road. To ensure vehicles can enter the road from upstream, we impose the boundary condition $u(0,t)=1$. Since this paper focuses on single-step control, for notational simplicity, $u(x)$ is used to denote $u(x,t)$ in the remainder of the paper.

\subsection{CBF Problem Description}

Before formulating the problem, we first introduce the CBF theory briefly for completeness. 

\emph{Control Barrier Functions}\cite{Ames2019Control}:
Denote a control system $\dot{x}=f(x,u)$ and a safe set defined by $h(x)\geqslant0$, where $h(x)$ is a continuously differentiable scalar function. Then, $h(x)$ is a CBF if there exists a selected extended class $\mathcal{K}_{\infty}$ function $\alpha$ such that the following inequality holds for any state $x$ 
\begin{equation}
    \underset{u}{\sup}~ \dot{h}(x,u) + \alpha(h(x)) \geqslant 0 
\end{equation}

For our purposes, the safety objective is to prevent congestion at the macroscopic level, thus the density $\rho$ should not exceed a certain value. Denoting the density bound by $\rho_{\text{limit}}$, the CBF and the corresponding constraint at a certain time $t_a$ can be built as
\begin{subequations}
    \begin{align}
        & \quad\quad\quad h(x) = \rho_{\text{limit}} - \rho(x,t_a) \geqslant 0
        \label{barrier function} \\
        & \dot{h}(x,u) + \alpha h(x) = -\frac{\partial\rho}{\partial t}\big|_{t=t_a} + \alpha h(x) \geqslant 0
        \label{original constraint}
    \end{align}
\end{subequations}
where we omit the time $t$ from the arguments of $h(x)$ for notational simplicity and the extended class $\mathcal{K}_{\infty}$ function is selected as a linear function with coefficient $\alpha$.

Combined with the LWR model \eqref{LWR model}, the CBF constraint \eqref{original constraint} can be recast, and we obtain the following pointwise-in-time infinite-dimensional optimization problem:
\begin{subequations}
	\begin{align}
		\underset{u(\cdot),u(0)=1}{\min} \,\, & \int_0^L \frac{1}{2}(u - u_{\text{ori}})^2\,\mathrm{d}x
		\label{cost}	\\
			\text{s.t.} \quad
			& \quad \mathcal{L}u -\alpha h + q \leqslant 0
        \label{CBF cons}	
	\end{align}
    \label{whole cbf}
\end{subequations}
where $\mathcal{L}$ is a linear operator satisfying $\mathcal{L}u = -(uQ)_x$. To get homogeneous boundary conditions, we substitute $u$ with $v$ which satisfies $v=u-1$ in the following analysis, and also denote $v_{\text{ori}}=u_{\text{ori}}-1$. Moreover, due to the unknown ramp flow, a slack variable $\xi$ is introduced to ensure the feasibility of the above optimization problem \eqref{whole cbf}. Accordingly, the problem can be reformulated as follows
\begin{subequations}
	\begin{align}
		\underset{v(\cdot),\xi(\cdot),v(0)=0}{\min} \, & \int_0^L J(v,\xi)\,\mathrm{d}x
		\label{slack cost}	\\
			\text{s.t.} \quad  
			& \mathcal{L}v + C - \xi\leqslant 0 
        \label{slack CBF cons} \\
            & \quad -\xi \leqslant 0
        \label{slack cons}	
	\end{align}
    \label{slack cbf}
\end{subequations}
where $J(v,\xi) = \frac{1}{2}(v-v_{\text{ori}})^2 + \frac{\beta}{2}\xi^2$ and $C = -\alpha h + q - Q_x$. The optimization problem \eqref{slack cbf} is solved at each time step to find the safe functional control input $v$ to keep the barrier function \eqref{barrier function} holding. Note that, at each time step, $\mathcal{L}$ and $C$ can be computed using current traffic information, and are irrelevant to the optimization variable $v$.

\section{Methodology}
\label{Section Methodology}

\subsection{Karush–Kuhn–Tucker Analysis}

According to Karush–Kuhn–Tucker (KKT) theory, we define $H(v,\xi)=J(v,\xi)+\lambda (\mathcal{L}v+C-\xi)-\mu \xi$, in which $\lambda$ and $\mu$ are non-negative functional multipliers w.r.t to the position $x$. Then, the optimal functional solution of the problem \eqref{slack cbf} satisfies the following KKT conditions
\begin{subequations}
    \begin{align}
        & \frac{\partial H}{\partial v} = v-v_{\text{ori}} + \mathcal{L}^*\lambda = 0 
        \label{KKT1}\\
        & \frac{\partial H}{\partial \xi} = \beta \xi -\lambda - \mu = 0 
        \label{KKT2}\\
        & \lambda (\mathcal{L}v+C-\xi) = 0, \lambda \geqslant 0, Lv+C-\xi \leqslant 0
        \label{KKT3}\\
        & \mu \xi = 0, \mu \geqslant 0, \xi \geqslant\ 0
        \label{KKT4}
    \end{align}
    \label{KKT}
\end{subequations}
where $\mathcal{L}^*$ is the adjoint operator of $\mathcal{L}$.

From \eqref{KKT2}, if the slack variable $\xi = 0$, we can infer that $\lambda = \mu = 0$, due to their positivity. So, we have $v = v_{\text{ori}}$ from \eqref{KKT1}, which means the CBF is not active. Otherwise, the CBF is active when the $\xi > 0$. When $\mu = 0$ in \eqref{KKT4}, after eliminating $\xi$, the KKT conditions \eqref{KKT} are equal to
\begin{subequations}
    \begin{align}
        & \quad v = v_{\text{ori}} - \mathcal{L}^*\lambda
        \label{KKT_A1}\\
        & \mathcal{L}v+C-\lambda/\beta = 0
        \label{KKT_A2}
    \end{align}
    \label{KKT_A}
\end{subequations}
Let $\mathcal{A}$, $\mathcal{I}\in[0, L]$ denote the space intervals where the CBF is active or not, respectively. Together with the above analysis, the optimization results are shown below 
\begin{equation}
    (v,\lambda) = 
    \begin{cases}
        (v_{\text{ori}},0),& x \in \mathcal{I}, \\
        \text{solution of \eqref{KKT_A}}, & x \in \mathcal{A}.
    \end{cases}
    \label{KKT results}
\end{equation}

\subsection{Discretization}
The active CBF intervals $\mathcal{A}$ need to be found to obtain the safe control input. As the variables $\rho, Q$ in the LWR model \eqref{LWR model} cannot be represented in analytical form, discretization is necessary to solve the infinite-dimensional optimization problem. It is worth emphasizing that, though the optimization problem \eqref{slack cbf} degenerates into a traditional vector-form optimization problem, which is a QP problem as well, the high dimension of the optimization variable $v$, ranges from tens to hundreds, depending on the discretization, makes it challenging to find a safe control input by classical solvers in real-time. Therefore, a tailored solution method is developed for the discretized problem.

We use Godunov's scheme to update the LWR model \eqref{LWR model}, assuming that there are $N$ cells in space, whose length is $\Delta x$. We use the subscript $i = 1,2,..., N$ to denote the component of the vector associated with the $i$-th cell. The supply and demand of each cell are given by
\begin{equation}
    S_i = 
    \begin{cases}
        f(\rho_c), & \rho_i < \rho_c \\
        f(\rho_i), & \rho_i \geqslant \rho_c
    \end{cases},\
    D_i = 
    \begin{cases}
        f(\rho_i), & \rho_i < \rho_c \\
        f(\rho_c), & \rho_i \geqslant \rho_c
    \end{cases}
    \label{SD framework}
\end{equation}
where $\rho_c$ is the critical density defined in \eqref{critical}. The flow at the interface between the $i$-th and $(i+1)$-th cells are
\begin{equation}
    Q_i = \min(D_i, S_{i+1})
\end{equation}
where the upstream flow of the first cell is given by $Q_0=\min(Q_{\text{up}}, S_1)$ and the downstream flow of the last cell is given by $Q_N=D_N$. Let control input $u_i$ adjust the flow $Q_i$, then the discrete update of LWR model \eqref{LWR model} is
\begin{equation}
    \begin{aligned}
        \frac{\partial\rho_i}{\partial t} &= q_i + \frac{u_{i-1}Q_{i-1}-u_iQ_i}{\Delta x} \\
        &= q_i + \frac{Q_{i-1}-Q_i}{\Delta x} + \frac{v_{i-1}Q_{i-1}-v_iQ_i}{\Delta x}
    \label{discrete update}
    \end{aligned}
\end{equation}
where $v_0=0$ as we claim in \cref{Section Model}, and we can infer that $C_i =  - \alpha h_i + q_i + \frac{Q_{i-1}-Q_i}{\Delta x}$. Therefore, the discrete form of the linear operator $\mathcal{L}$, which is denoted by $K$, satisfies
\begin{equation}
    K = \left(\frac{k_{ij}}{\Delta x}\right), \quad
    k_{ij} =
    \begin{cases}
    -Q_i, & j = i, \\
    Q_i,  & j = i-1, \\
    0,    & \text{otherwise}.
    \end{cases}
    \label{K}
\end{equation}
where $K$ is a $N\times N$ lower bi-diagonal matrix and $K^T$ is the discrete form of the adjoint operator $\mathcal{L}^*$. We would like to emphasize that the bi-diagonal structure of $K$ is provided by the physical properties of the traffic flow, which is important for the convergence of the proposed algorithm.

\subsection{Controller Design Algorithm}

From the activity discussion of $\xi$ in \eqref{KKT4}, we find that $\xi = \lambda/\beta$ always holds. Together with \eqref{KKT results}, the KKT conditions \eqref{KKT} are equal to finding roots of the following function
\begin{equation}
    F(v,\lambda) = 
    \begin{pmatrix}
        v - v_{ori} + K^T\lambda \\
        \lambda - \max(0,g(v,\lambda))
    \end{pmatrix}
    \label{equal_eqn}
\end{equation}
where $g(v,\lambda)=\lambda + d(Kv+C-\lambda/\beta)$ and $d$ is a positive constant. In theory, the choice of $d$ does not affect the result; however, in practical applications, an appropriate $d$ can help avoid numerical issues. We use subscripts $\mathcal{A}$ and $\mathcal{I}$ to partition the matrices or vectors on active and inactive sets. By eliminating $v$ in \eqref{KKT_A2}, on the active set $\mathcal{A}$, $\lambda$ satisfies 
\begin{equation}
    ((KK^T+I_N/\beta)\lambda)_{\mathcal{A}} = (Kv_{\text{ori}}+C)_{\mathcal{A}}
    \label{KKT_A3}
\end{equation}
where $I_N$ is a $N\times N$ identity matrix. Inspired by \cite{Hintermuller2002PrimalDual}, we propose \cref{alg:PDAS}, a PDAS-based iteration, to obtain the active set $\mathcal{A}$ and optimal control input $v$. The convergence of \cref{alg:PDAS} is proved by the following corollary:
    
\begin{algorithm}[!t] 
    \caption{PDAS-based CBF Controller} 
	\renewcommand{\algorithmicrequire}{\textbf{Input:}}
	\renewcommand{\algorithmicensure}{\textbf{Output:}}
    \begin{algorithmic}[1] 
    	\REQUIRE
    		Matrix $K$; Vectors $C,\beta$; Positive constant $d$; Initial multiplier $\lambda^0$ and responding initial control input $v^0 = v_{\text{ori}} - K^T\lambda^{0}$(e.g., $\lambda^0=0$ and $v^0 = v_{\text{ori}}$). 
    	\ENSURE
    		Safe control input $v$.
        \STATE $k \gets 0$.
        \WHILE{TRUE}
            \STATE
        	    Calculate the  active set $\mathcal{A}_k=\{i \mid g_i(v^k,\lambda^k)>0\}$ and inactive set $\mathcal{I}_k=\{i \mid g_i(v^k,\lambda^k)\leqslant0\}$.
            \STATE
                Set $\lambda^{k+1}_{\mathcal{I}_k}=0$.
            \STATE
                Solve $\lambda^{k+1}_{\mathcal{A}_k}$ by equation \eqref{KKT_A3}.
            \STATE
                Calculate $v^{k+1}$ by $v_{ori} - K^T\lambda^{k+1}$.
            \IF{$v^{k+1} = v^k$}
                \STATE \textbf{break}   
            \ELSE
                \STATE $k \gets k + 1$.
            \ENDIF     
        \ENDWHILE
        \STATE $v = v^{k+1}$.
    \end{algorithmic} 
	\label{alg:PDAS} 
\end{algorithm}

\begin{corollary}
	Let $(v^*,\lambda^*)$ denote the unique solution of function $F(v,\lambda)$ in \eqref{equal_eqn}. For arbitrary initial multiplier $\lambda^0$ and vector $v^0 = v_{\text{ori}}-K^T\lambda^0$, we have $\lambda^k\leqslant\lambda^{k+1}\leqslant\lambda^*$ for all $k\geqslant1$ and  $\lambda^k\geqslant0$ for all $k\geqslant2$.
	\label{corollary convergence}
\end{corollary}

\begin{proof}
    The proof is provided in the Appendix. 
\end{proof}

Together with the non-negativity and monotonicity of $\lambda^k$ for all $k\geqslant2$, once index $i$ is identified in the active set, it will not be classified as inactive in subsequent iterations. \cref{alg:PDAS} is terminated after a finite number of iterations, not exceeding a predefined maximum, and $(v^*,\lambda^*)$ is obtained.

\begin{remark}
Indeed, the operator $\mathcal{L}\mathcal{L}^*+\mathcal{I}_d/\beta$ is an elliptic operator, where $\mathcal{I}_d$ is the identity operator.The operator shares similar positivity properties with the matrix considered in the Appendix. We will investigate the theoretical analysis in continuous spatial domain in future work.
\end{remark}

\section{Simulation Results}
\label{Section Simulation}

\subsection{Simulation Settings}

The model parameters used in the simulation are shown in Table~\ref{Table:simulation_parameters}, where the parameter values are selected based on real-world data \cite{Block2025LQInformed}. 
The Greenshield model, the selected fundamental diagram $f(\rho)$, and the corresponding critical density $\rho_c$ are given by
\begin{equation}
        Q = f(\rho) = \rho V_{\text{max}}\bigg(1-\frac{\rho}{\rho_{\text{max}}}\bigg), \ 
        \rho_c = \frac{\rho_{\text{max}}}{2}
\end{equation}
The upstream condition $Q_{\text{{up}}}$, initial density $\rho^0$, and ramp flow $q$ of the case study are
\begin{subequations}
    \begin{align}
        & Q_{\text{{up}}}(t) = 45+5\exp(-10^{-5}\pi tL)\cdot\sin\bigg(\frac{\pi t}{40}\bigg)+\frac{20 t}{L} \\
        & \quad\quad\quad\quad\quad \rho^0(x) = 45+10\cdot\sin\bigg(\frac{\pi x}{L}\bigg)
        \label{initial free}\\
        & q(x,t) = 
        \begin{cases}
            1.2+0.3\sin(\frac{\pi t}{25}),& \text{see Table ~\ref{Table:simulation_parameters},}  \\
            0, & \text{otherwise.}
        \end{cases}
        \label{inflow}
    \end{align}
\end{subequations}
From \eqref{initial free}, the maximum initial density is less than the critical density $\rho_c=60$ veh/km, indicating that the initial condition is in free-flow and congestion observed in the case study is caused by on-ramp flow. The average ramp flow in \eqref{inflow} is $120$ veh/hr and does not exceed $500$ veh/hr, which is consistent with real-world data \cite{Block2025LQInformed}.

\begin{table}[!t]
	\centering
	\caption{Simulation Parameters}
	\begin{tabular}{c c c}
		\toprule
		Parameter & Unit & Value \\
		\midrule
        Maximum density $\rho_{\text{max}}$ & veh/km & 120 \\
		Maximum velocity $V_{\text{max}}$ & km/h & 72 \\
        Road length $L$ & m & 400 \\
        Simulation time $T$ & s & 250 \\
        Ramp location $x_q$ & m & $100 \leqslant x_q \leqslant 200$ \\
        Ramp time $t_q$ & s & $0 \leqslant t_q \leqslant 200$ \\
		\bottomrule
	\end{tabular}
	\label{Table:simulation_parameters}
\end{table}

We consider two different nominal controllers: a zero-input (no velocity regulation) and an LQ controller from \cite{Block2024LQ}. We briefly introduce the selected LQ controller. For the desired traffic density $\rho_s$, the linearized LWR model near $(\rho_s,v_s=0)$ without ramp flow can be expressed as 
\begin{equation}
    \frac{\partial\Delta\rho}{\partial t}
    = -V_{\text{max}}\bigg(1-\frac{2\rho_s}{\rho_{\text{max}}}\bigg)\frac{\partial\Delta\rho}{\partial x} - \rho_s V_{\text{max}}\bigg(1-\frac{\rho_s}{\rho_{\text{max}}}\bigg)\frac{\partial\Delta v}{\partial x}
    \label{LQ_system}
\end{equation}
where $\rho = \rho_s+\Delta\rho$ and $v = v_s+\Delta v=\Delta v$. Let $V=-V_{\text{max}}(1-\frac{2\rho_s}{\rho_{\text{max}}})$ and $B_s = - \rho_s V_{\text{max}}(1-\frac{\rho_s}{\rho_{\text{max}}})$, and we denote $y = \Delta\rho$ and $z = \frac{\partial v}{\partial x}$ (note that $v_s=0$ and $v = \Delta v$), then the equivalent state space formulation of the linearized LWR model \eqref{LQ_system} on a Hilbert space is
\begin{equation}
    \dot{y} (t) = V\frac{\mathrm{d}}{\mathrm{d}x}y(t) + B_s z(t)
    \label{Hilbert system}
\end{equation}
where $z(t)$ is the rate of change of the input. Given a positive constant $P_s$, to minimize the cost $\int_0^{\infty}P_s\lvert y \rvert^2 + \lvert z \rvert^2\,\mathrm{d}t$ for the system \eqref{Hilbert system}, the optimal control input is 
\begin{equation}
    z = \frac{\partial v}{\partial x} = -\sqrt{P_s}\frac{\exp(\frac{2B_s\sqrt{P_s}}{V}(x-L))-1}{\exp(\frac{2B_s\sqrt{P_s}}{V}(x-L))+1}\Delta\rho
    \label{LQ controller}
\end{equation}
and the original $v$ can be computed accordingly. For details and properties of the LQ controller \eqref{LQ controller}, refer to \cite{Block2024LQ}. 

The controller parameters used in the simulation, together with their physical units, are shown in Table~\ref{Table:controller_parameters}. 
\begin{table}[!t]
	\centering
	\caption{Controller Parameters}
	\begin{tabular}{c c c}
		\toprule
		Parameter & Unit & Value \\
		\midrule
        CBF density limit $\rho_{\text{limit}}$ & veh/km & 60 \\
        LQ target density $\rho_s$ & veh/km & 55 \\
        CBF linear function coefficient $\alpha$ & s$^{-1}$ & $0.8$ \\
        CBF penalty coefficient $\beta$ & m$^2\cdot$ s$^2$/veh$^2$ & $1\cdot10^{4}$ \\
        LQ tuning $P_s$ & veh$^{-2}$ & $5\cdot10^{-3}$ \\
        Positive constant $d$ in $g(v,\lambda)$ & m$^2\cdot$ s$^2$/veh$^2$ & $2\cdot10^{4}$ \\
		\bottomrule
	\end{tabular}
	\label{Table:controller_parameters}
\end{table}
The CBF density limit $\rho_{\text{limit}}$ is chosen to be equal to the critical density $\rho_c$ to prevent congestion. A relatively large CBF penalty coefficient $\beta$ is assigned to the slack variable, resulting in smaller activated slack values. As ramp flows are not modeled in the original LQ controller design, compared with \cite{Block2024LQ}, we selected a more aggressive tuning parameter $P_s$ to achieve stronger regulation capability.

\subsection{Performance Evaluation}

The simulation results of the traffic density, velocity, and speed limit are shown in Fig.~\ref{fig:density}, Fig.~\ref{fig:velocity}, and Fig.~\ref{fig:vsl}, respectively. Note that Fig.~\ref{fig:density}(\subref{fig:density_1}) and Fig.~\ref{fig:velocity}(\subref{fig:velocity_1}) show the traffic evolution in the absence of any control input under the given parameters. Due to the continuous ramp flow, congestion emerges and propagates upstream as expected.

When the CBF is applied with the zero-input controller, no congestion is observed in Fig.~\ref{fig:density}(\subref{fig:density_2}). Due to the continuous inflow from the ramp, the active region of the CBF expands to contain more vehicles, until it is activated along the entire road at around $180$ s to $200$ s. Moreover, from Fig.~\ref{fig:velocity}(\subref{fig:velocity_2}), the upstream traffic velocity is always slower than the downstream inside the CBF active region, meaning that the CBF tries to prevent vehicles from entering the active region and encourages them to exit in order to avoid congestion. 

\begin{figure}[!t]
    \centering
    \begin{subfigure}[b]{0.45\linewidth}
        \includegraphics[width=\linewidth]{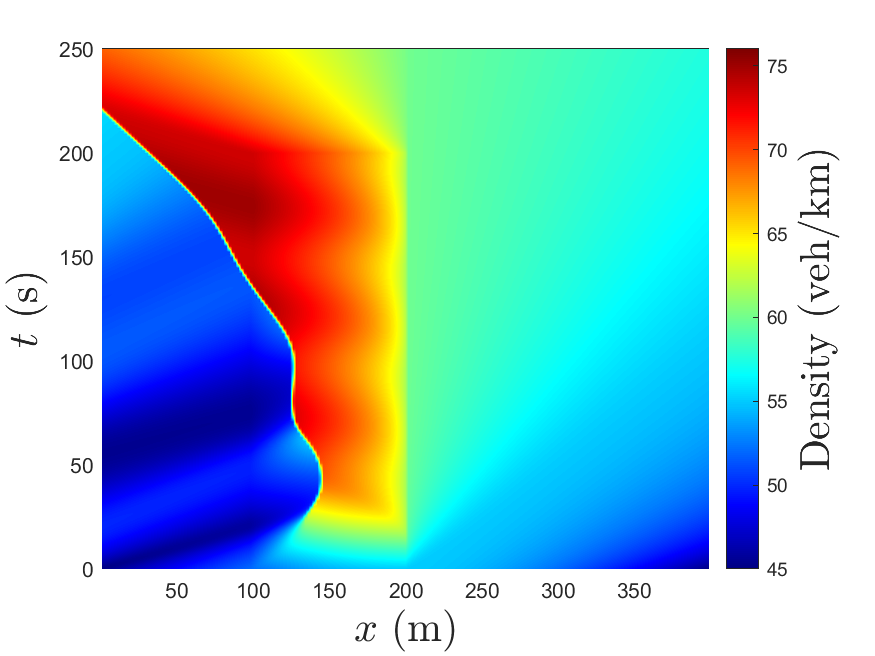}
        \subcaption{}
        \label{fig:density_1}
    \end{subfigure}
    \hfil
    \begin{subfigure}[b]{0.45\linewidth}
        \includegraphics[width=\linewidth]{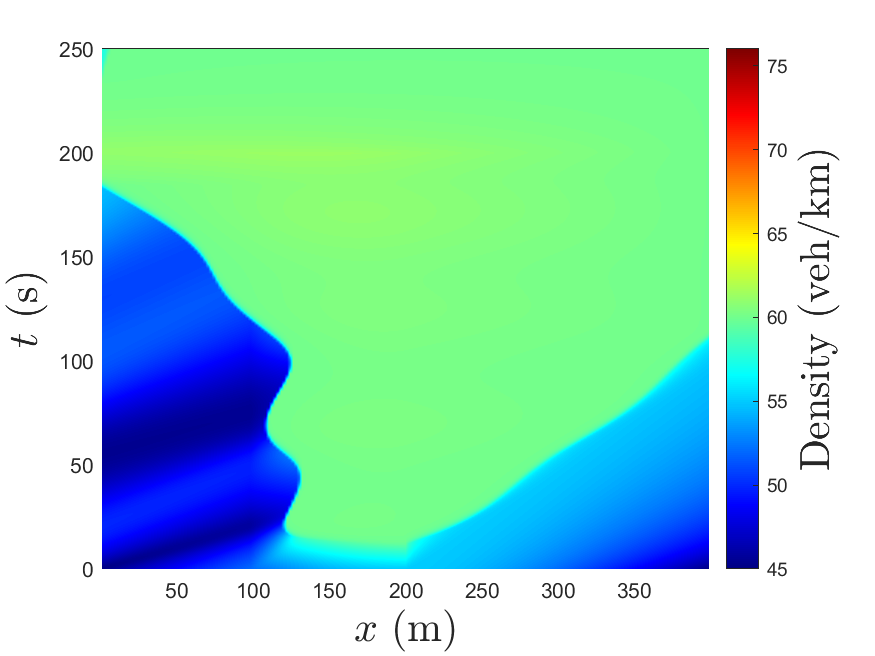}
        \caption{}
        \label{fig:density_2}
    \end{subfigure}
    \vspace{0.2cm}
    \begin{subfigure}[b]{0.45\linewidth}
        \includegraphics[width=\linewidth]{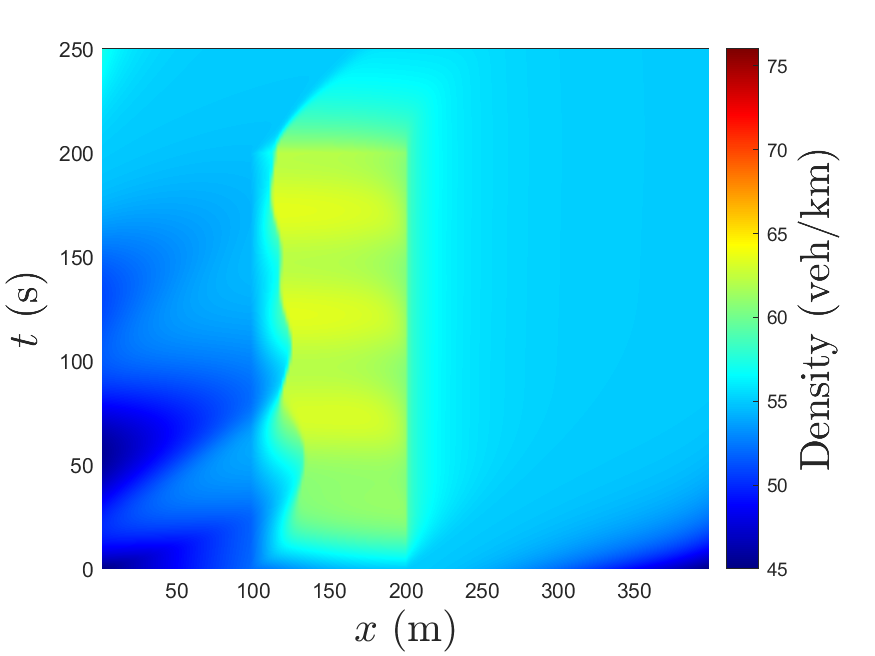}
        \caption{}
        \label{fig:density_3}
    \end{subfigure}
    \hfil
    \begin{subfigure}[b]{0.45\linewidth}
        \includegraphics[width=\linewidth]{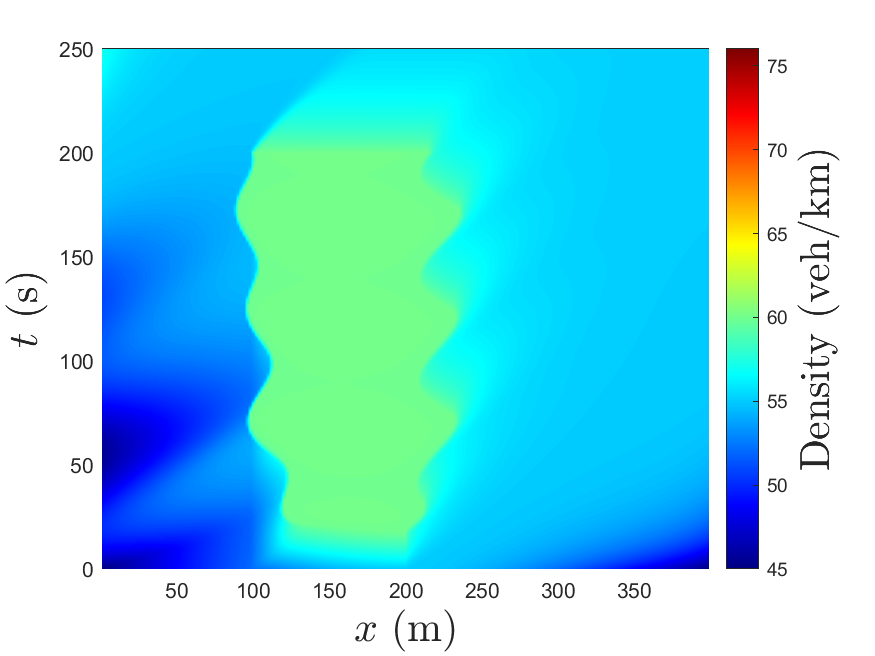}
        \caption{}
        \label{fig:density_4}
    \end{subfigure}
    \caption{Traffic density of (a)  zero-input without CBF , (b)  zero-input with CBF, (c)  LQ controller without CBF, and (d) LQ controller with CBF.}
    \label{fig:density}
\end{figure}

\begin{figure}[!t]
    \centering
    \begin{subfigure}[b]{0.45\linewidth}
        \includegraphics[width=\linewidth]{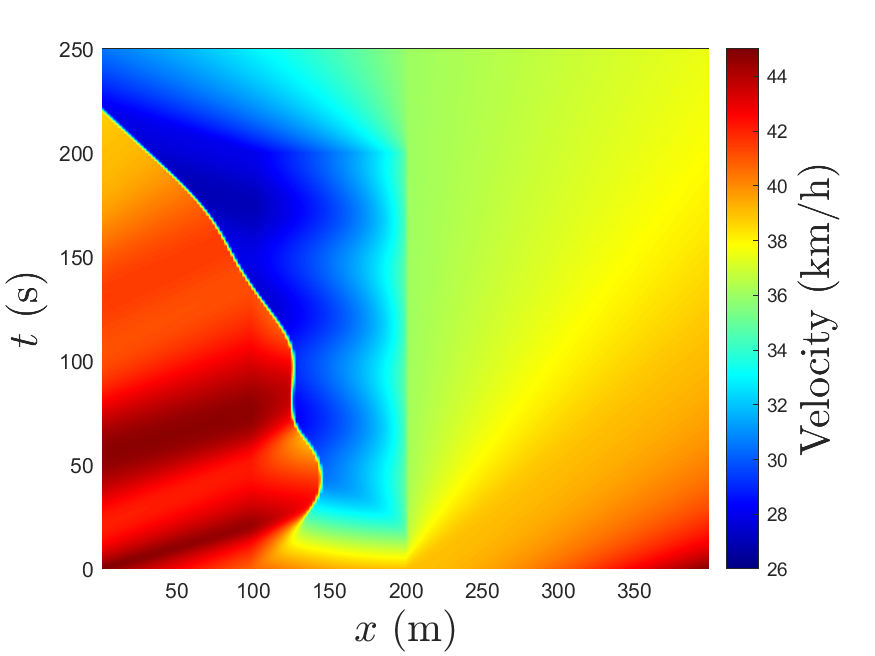}
        \caption{}
        \label{fig:velocity_1}
    \end{subfigure}
    \hfil
    \begin{subfigure}[b]{0.45\linewidth}
        \includegraphics[width=\linewidth]{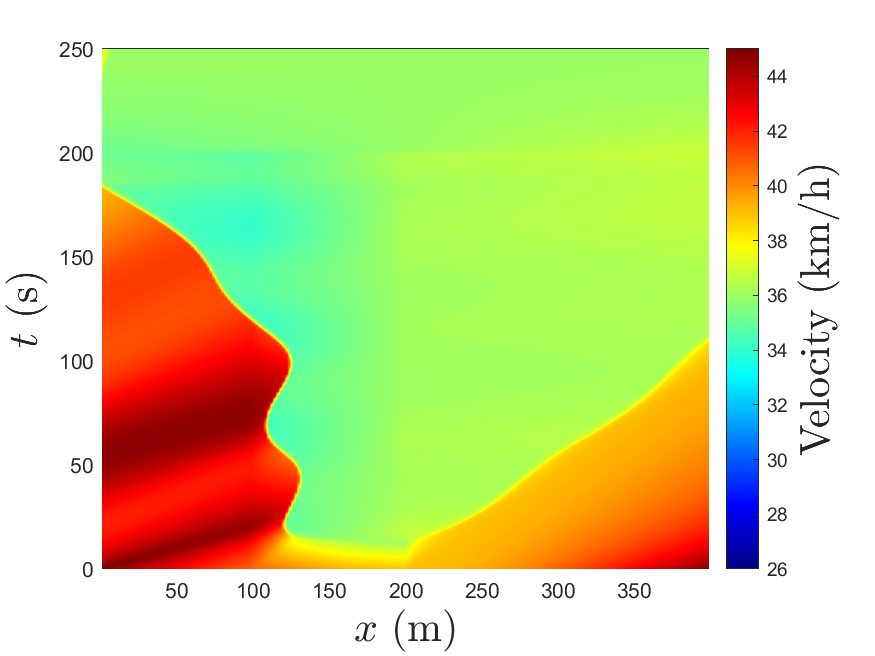}
        \caption{}
        \label{fig:velocity_2}
    \end{subfigure}
    \vspace{0.2cm}
    \begin{subfigure}[b]{0.45\linewidth}
        \includegraphics[width=\linewidth]{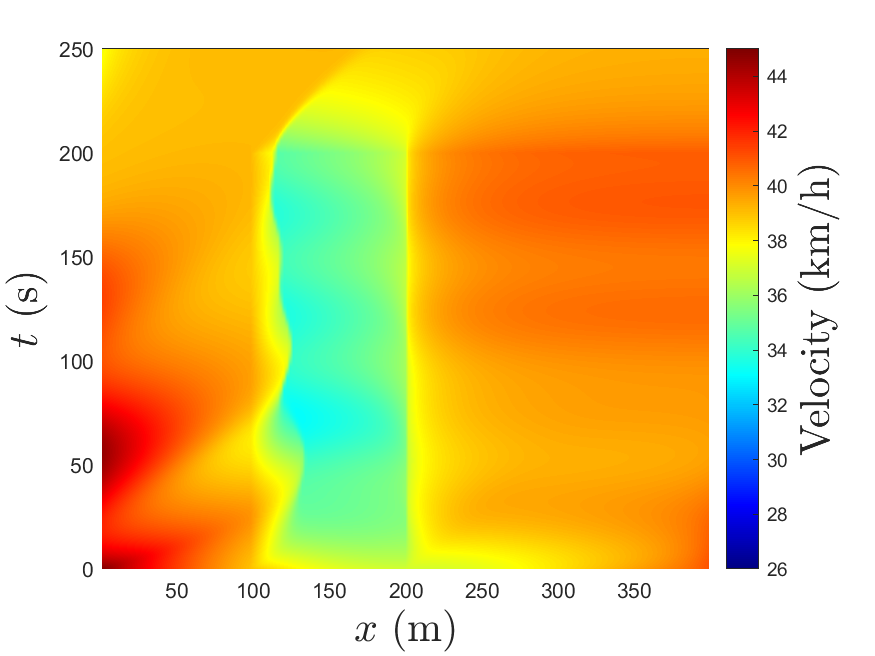}
        \caption{}
        \label{fig:velocity_3}
    \end{subfigure}
    \hfil
    \begin{subfigure}[b]{0.45\linewidth}
        \includegraphics[width=\linewidth]{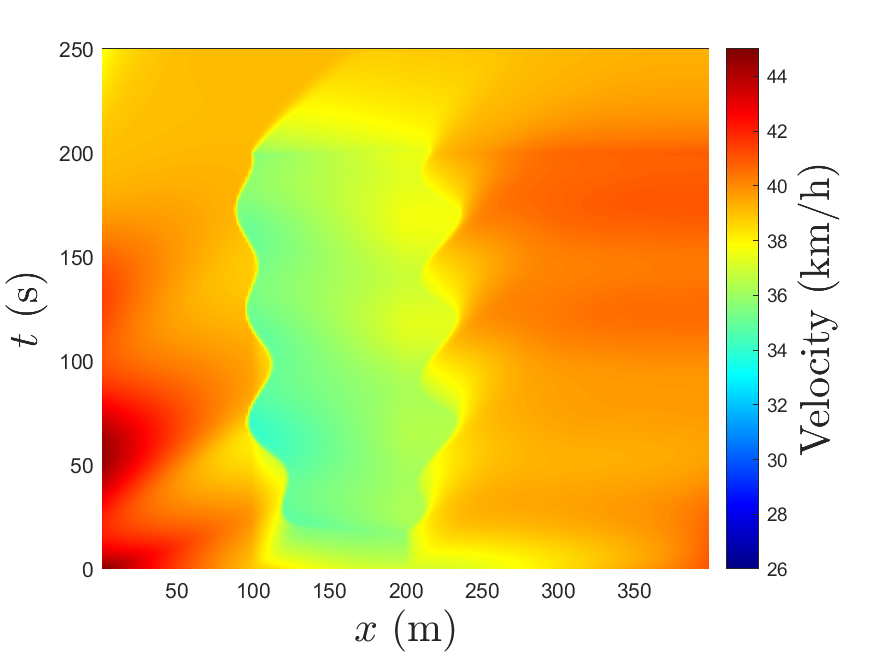}
        \caption{}
        \label{fig:velocity_4}
    \end{subfigure}
    \caption{Traffic velocity of (a)  zero-input without CBF, (b)  zero-input with CBF, (c)  LQ controller without CBF, and (d)  LQ controller with CBF.}
    \label{fig:velocity}
\end{figure}

\begin{figure}[!t]
    \centering
    \begin{subfigure}[b]{0.45\linewidth}
        \includegraphics[width=\linewidth]{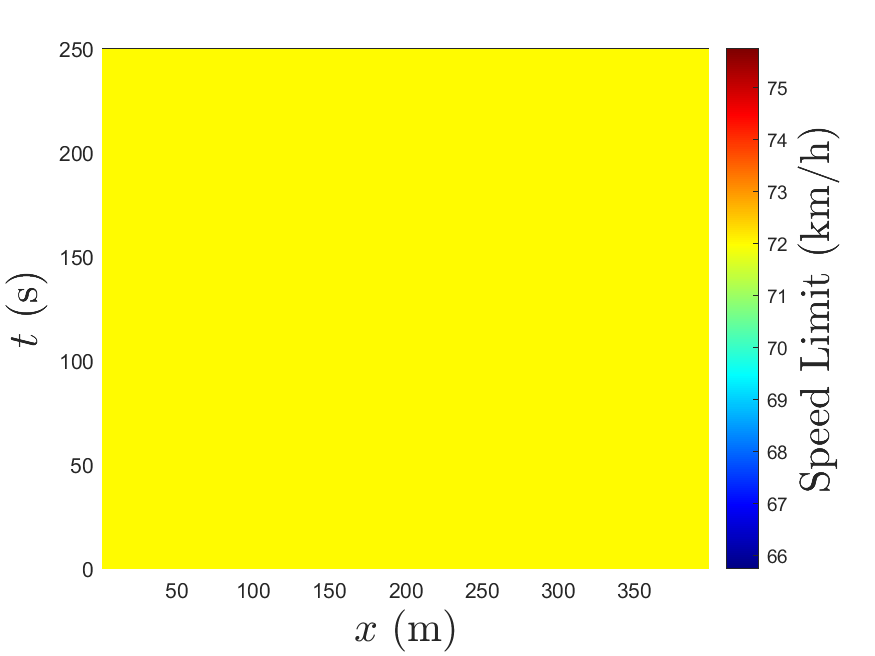}
        \caption{}
        \label{fig:vsl_1}
    \end{subfigure}
    \hfil
    \begin{subfigure}[b]{0.45\linewidth}
        \includegraphics[width=\linewidth]{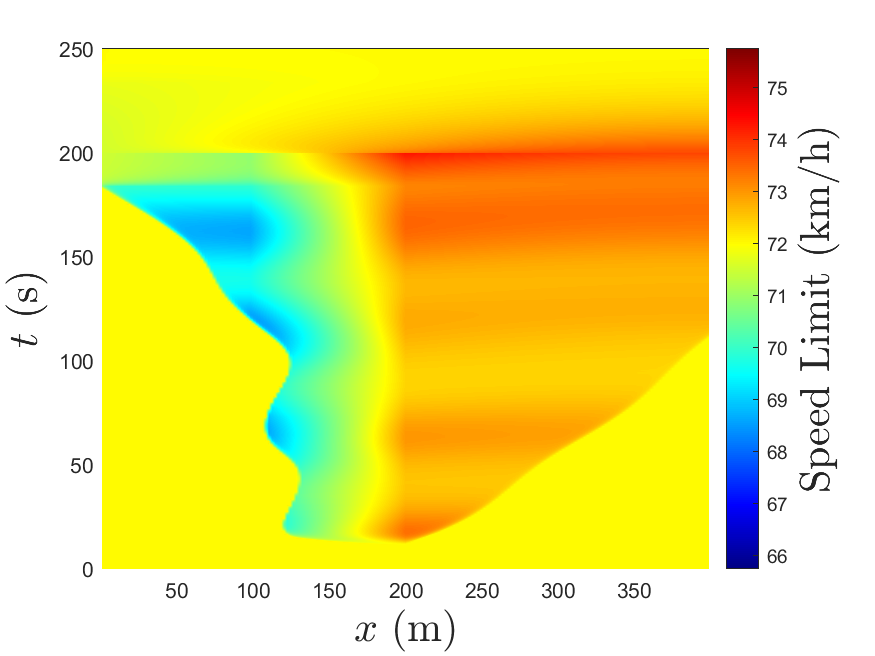}
        \caption{}
        \label{fig:vsl_2}
    \end{subfigure}
    \vspace{0.2cm}
    \begin{subfigure}[b]{0.45\linewidth}
        \includegraphics[width=\linewidth]{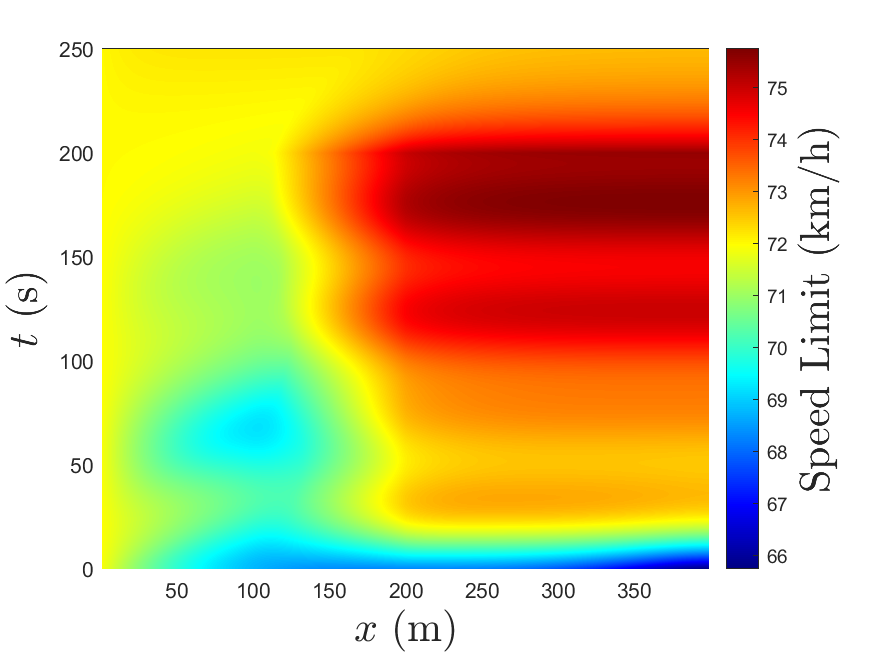}
        \caption{}
        \label{fig:vsl_3}
    \end{subfigure}
    \hfil
    \begin{subfigure}[b]{0.45\linewidth}
        \includegraphics[width=\linewidth]{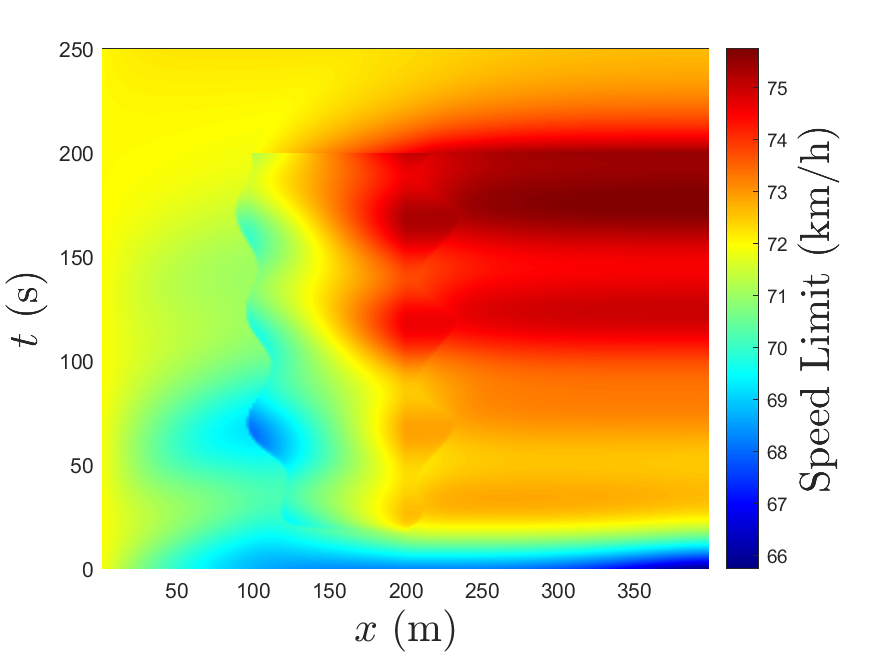}
        \caption{}
        \label{fig:vsl_4}
    \end{subfigure}
    \caption{Speed limit of (a)  zero-input without CBF, (b)  zero-input with CBF, (c)  LQ controller without CBF, and (d)  LQ controller with CBF.}
    \label{fig:vsl}
\end{figure}

On the other hand, the regulation capability of the LQ controller suppresses congestion as shown in Fig.~\ref{fig:density}(\subref{fig:density_3}). By regulating the velocity of the congested regime in Fig.~\ref{fig:velocity}(\subref{fig:velocity_3}), upstream propagation is not observed. After $200$ s, with no further on-ramp flow, the traffic density is driven to the desired $\rho_s$. However, as the ramp condition is not considered during the LQ controller design, the traffic density exhibits an overshoot due to the on-ramp flow, and the magnitude of the overshoot increases as the inflow rate increases. Because of this, slight congestion can be observed at the ramp location. 

When the combination of the LQ controller and CBF is applied, no congested regime is shown in Fig.~\ref{fig:density}(\subref{fig:density_4}). Moreover, the regulation capability of the LQ controller makes the active region of the CBF not expand continuously. Notably in Fig.~\ref{fig:velocity}(\subref{fig:velocity_4}), in the downstream region of the ramp, traffic velocity is higher over a wider area. After $200$ s, once the ramp inflow ceases, the CBF is not activated and the LQ controller drives the system to the desired $\rho_s$.

We also quantitatively evaluate the controller performance using two metrics, total time spent (TTS) and total delay (TD), whose definitions are
\begin{subequations}
    \begin{align}
        & \quad\quad\quad\quad \text{TTS} = \int_0^{T_q}\int_0^{L} \rho(x,t)\,\mathrm{d}x\,\mathrm{d}t \\
        & \text{TD} = \int_0^{T_q}\int_0^{L} \rho(x,t) (V_{\text{max}}-V(x,t))\,\mathrm{d}x\,\mathrm{d}t
    \end{align}
\end{subequations}
where $T_q=200$ represents the time the on-ramp inflow stops. TTS measures cumulative time spent by all vehicles in the system, and TD quantifies the additional time incurred compared to free-flow conditions. The quantitative results are shown in Table~\ref{Table:results}. 
\begin{table}[!t]
    \centering
    \caption{Performance under Different Simulation Cases}
    \begin{tabular}{ccc}
    \toprule
    Simulation Case & TTS~[veh$\cdot$h] & TD~[veh$\cdot$km] \\
    \midrule
        Zero-input w/o CBF & 1.2725 & 44.4322 \\
        Zero-input w/ CBF  & 1.2610 & 43.2255 \\
        LQ controller w/o CBF & 1.2379 & 41.2342 \\
        LQ controller w/ CBF  & 1.2379 &  41.2285 \\
    \bottomrule
    \end{tabular}
    \label{Table:results}
\end{table}
Lower TTS and TD indicate that the proposed CBF-based method improves traffic efficiency, with significant gains in the zero-input case. When combined with an LQ controller, performance remains essentially unchanged despite the aggressive tuning parameter $P_s$, indicating that the safety filter preserves throughput while enforcing safety.

Additional simulations were conducted with initially congested conditions. Similar behavior is observed, where the proposed method removes congestion while preserving the performance trends reported in Table~\ref{Table:results}.

\subsection{Real-time Feasibility}

We compare the proposed method with \texttt{quadprog} and \texttt{osqp} in MATLAB to showcase its real-time feasibility. $(v_{\text{ori}}, \lambda=0)$ is set as the warm start for all methods 
. Results are shown in Fig.~\ref{fig:time_analysis}, where the time ratio is defined as the computation time divided by the simulation time step.

\begin{figure}[!t]
    \centering
    \begin{subfigure}[b]{0.45\linewidth}
        \includegraphics[width=\linewidth]{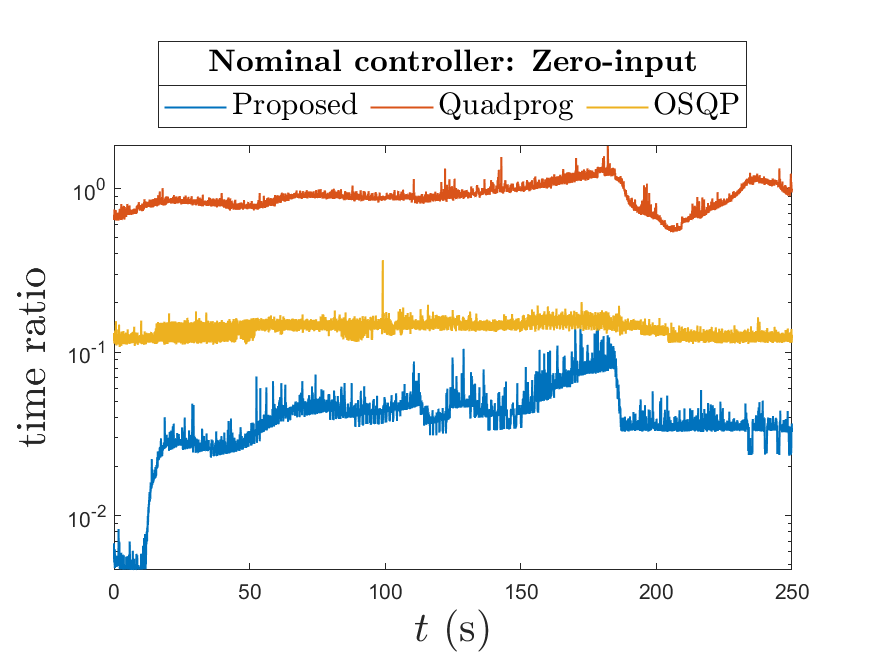}
        \caption{}
        \label{fig:time_1}
    \end{subfigure}
    \hfil
    \begin{subfigure}[b]{0.45\linewidth}
        \includegraphics[width=\linewidth]{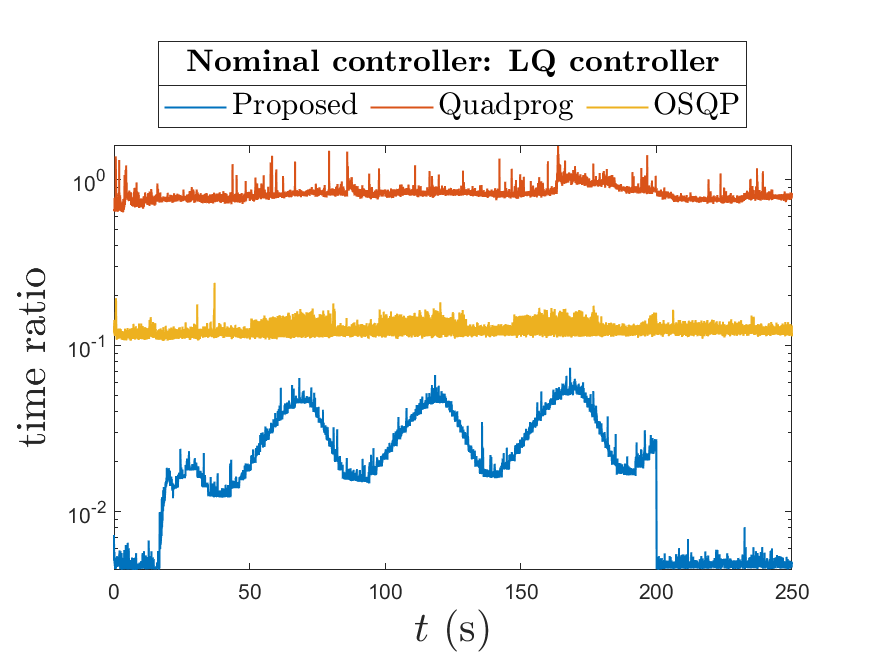}
        \caption{}
        \label{fig:time_2}
    \end{subfigure}
    \caption{Computation time comparison with different methods for nominal controllers: (a)  zero-input and (b) LQ controller.}
    \label{fig:time_analysis}
\end{figure}

The proposed method presents better computational efficiency than classical solvers for both nominal controllers, which is beneficial for further applications. By comparing the results in Fig.\ref{fig:time_analysis}(\subref{fig:time_1}) and  Fig.\ref{fig:time_analysis}(\subref{fig:time_2}), it can be concluded that a proper nominal controller is helpful to improve the computational efficiency of the proposed method. As we mentioned before, compared with the zero-input, the LQ controller can reduce the active region of the CBF. This leads to fewer iterations in \cref{alg:PDAS}, thus less computational time. The oscillations in computation time for the proposed method in Fig.\ref{fig:time_analysis}(\subref{fig:time_2}) come from the sinusoidal ramp inflow and the changing width of the active CBF set.


\section{Conclusion}
\label{Section Conclusion}

In this paper, a CBF-based method for real-time in-domain safe control of the LWR traffic model was proposed. The corresponding infinite-dimensional optimization problem was reformulated through KKT analysis and discretization, leading to a structured quadratic program. A PDAS-based algorithm was developed to efficiently compute the pointwise-in-time safe control input, with convergence guarantees. Simulation results demonstrated that congestion can be effectively prevented using CBF and that when combined with an LQ controller, the performance remains largely unaffected. In addition, improved computational efficiency is achieved compared to classical QP solvers. Future work will focus on the following directions: (i) theoretical analysis of the proposed safety filter in the continuous spatial domain and extension to higher-dimensional models; (ii) improvement of the iteration policy, incorporation of practical speed constraints, and feasibility analysis under different parameter settings.





\section*{APPENDIX}
\section*{Proof of Corollary~\ref{corollary convergence}}

\begin{proof}
    Let $M=KK^T+I_N/\beta$, it is easy to show that $M$ is an M-matrix from \eqref{K}. Note that $v^{k} = v_{\text{ori}} - K^T\lambda^{k}$ always holds, and as $M$ is an M-matrix, we have $M_{\mathcal{A}\mathcal{A}}^{-1}>0$ (denoted by $M_{\mathcal{A}}^{-1}$ later) and $M_{\mathcal{A}}^{-1}M_{\mathcal{A}\mathcal{I}}<0$ for any partition $\mathcal{A}$ and $\mathcal{I}$. Firstly, we show that
    \begin{equation}
        \lambda^k_i \neq 0 \implies (Kv^k+C-\lambda^k/\beta)_i = 0, \, \forall k\geqslant1
        \label{proof1}
    \end{equation}
    For $k\geqslant1$, as $ \lambda^k_{\mathcal{I}_{k-1}} = 0$ holds in each iteration, we have $i\in{\mathcal{A}_{k-1}}$ if $\lambda^k_i \neq 0$ and equation \eqref{KKT_A3} holds. This implies
    \begin{equation}
        \begin{split}
            &(Kv_{\text{ori}}+C-(KK^T+I_N/\beta)\lambda^k)_{\mathcal{A}_{k-1}} \\
            &\quad = (Kv^k+C-\lambda^k/\beta)_{\mathcal{A}_{k-1}}=0
        \end{split}
        \label{proof1_1}
    \end{equation}
    which can lead to $(Kv^k+C-\lambda^k/\beta)_i = 0$ if $\lambda^k_i \neq 0$.

    Next, we check the monotonicity of $\lambda^k$ for $k \geqslant 1$. For $i\in\mathcal{I}_k$, if $\lambda^k_i>0$, then we have $g_i(v^k,\lambda^k)=\lambda^k_i + d(Kv^k+C-\lambda^k/\beta)_i>0$ by \eqref{proof1}, which is contradictory. This shows that $\lambda^k_i\leqslant0$ for $i\in\mathcal{I}_k$, and implies 
    \begin{equation}
    \lambda^{k+1}_{\mathcal{I}_k}= 0 \geqslant\lambda^k_{\mathcal{I}_k}
    \label{proof2_1}
    \end{equation}
    For $i\in\mathcal{A}_k$, we have $g_i(v^k,\lambda^k)=\lambda^k_i + d(Kv^k+C-\lambda^k/\beta)_i > 0$, hence either $\lambda^k_i > 0$ implying $(Kv^k+C-\lambda^k/\beta)_i = 0$ by \eqref{proof1}, or $(Kv^k+C-\lambda^k/\beta)_i > 0$. Therefore, by \eqref{KKT_A3}, $(M\lambda^{k+1})_{\mathcal{A}_k} = (Kv_{\text{ori}}+C)_{\mathcal{A}_k}=(M\lambda^k+Kv^k+C-\lambda^k/\beta)_{\mathcal{A}_k}\geqslant(M\lambda^k)_{\mathcal{A}_k}$, and we have
    \begin{equation}
    \begin{split}
        &(M(\lambda^{k+1}-\lambda^k))_{\mathcal{A}_k} = M_{\mathcal{A}_k}\big(
        (\lambda^{k+1}-\lambda^k)_{\mathcal{A}_k} \\
        &\quad + M_{\mathcal{A}_k}^{-1} M_{\mathcal{A}_k\mathcal{I}_k}(\lambda^{k+1}-\lambda^k)_{\mathcal{I}_k} \big) \geqslant0
    \end{split}
        \label{proof2_2}
    \end{equation}
    which follows that $\lambda^{k+1}_{\mathcal{A}_k} \geqslant \lambda^k_{\mathcal{A}_k}$, by combining \eqref{proof2_1} and the M-matrix properties. Thus, the monotonicity of $\lambda^k$ is shown.

    Then, we prove the non-negativity of $\lambda^k$ for $k\geqslant2$. For any $i$ satisfying $\lambda^1_i<0$, we have $g_i(v^1,\lambda^1)<0$ by \eqref{proof1}, thus $i\in\mathcal{I}_1$ and $\lambda^2_i=0$. Combined with the monotonicity of $\lambda^k$, $\lambda^k\geqslant\lambda^2\geqslant0$ for all $k\geqslant2$.

    Finally, we verify that $\lambda^*\geqslant\lambda^k$ for $k\geqslant1$. For inactive set $\mathcal{I}_{k-1}$, $\lambda^*_{\mathcal{I}_{k-1}}\geqslant0=\lambda^k_{\mathcal{I}_{k-1}}$ . For active set $\mathcal{A}_{k-1}$, note the definition of $(v^*,\lambda^*)$, we have $\lambda^*= \max(0,g(v^*,\lambda^*))=\max(0,\lambda^*+d(Kv_{\text{ori}}+C-M\lambda^*))$. As equation \eqref{KKT_A3} holds, we infer that $\lambda^*_{\mathcal{A}_{k-1}}=\max(0,\lambda^*_{\mathcal{A}_{k-1}}+d(M(\lambda^k-\lambda^*))_{\mathcal{A}_{k-1}})$, which implies that $(M(\lambda^k-\lambda^*))_{\mathcal{A}_{k-1}} < 0$. Similarly to \eqref{proof2_2}, we have $\lambda^*\geqslant\lambda^k$ for $k\geqslant1$.
\end{proof}


\bibliographystyle{Bibliography/IEEEtran}
\bibliography{Bibliography/IEEEabrv, Bibliography/mybib}

\end{document}